\documentclass[12pt]{amsart}
\usepackage{latexsym,fancyhdr,amssymb,color,amsmath,amsthm,graphicx,listings,comment}
\usepackage[section]{placeins}
\newtheorem{thm}{Theorem}

\let\paragraph\subsection

\title{On Symmetries of Finite Geometries}
\author{Oliver Knill}
\date{August 31, 2024}
\address{Department of Mathematics \\ Harvard University \\ Cambridge, MA, 02138 }


\begin{document}
\maketitle

\begin{abstract}
The isospectral set of the Dirac matrix $D=d+d^*$
consists of orthogonal $Q$ for which $Q^* D Q$ is an equivalent Dirac matrix. 
It can serve as the symmetry of a finite geometry $G$.
The symmetry is a subset of the orthogonal group or unitary 
group and isospectral Lax deformations produce commuting flows 
$d/dt D=[B(g(D)),D]$ on this symmetry space. In this note, we
remark that like in the Toda case, 
$D_t=Q_t^* D_0 Q_t$ with $e^{-t g(D)}=Q_t R_t$ solves the 
Lax system. 
\end{abstract}

\section{Finite Geometries}

\paragraph{}
Let $G$ be a finite set with $n$ elements. 
A {\bf dimension function} $R:G \to \{0,1,\dots,q\}$
defines a partition $G=\bigcup_{k=0}^q G_k$ of $G$. 
A {\bf Dirac matrix} is a symmetric {\bf block-tri-diagonal} 
$n \times n$ matrix $D=d+d^*+m$ such 
that $L=D^2$ and $C=(d+d^*)^2$ are both {\bf block diagonal} with respect 
to the decomposition $l^2(G)=\oplus_k l^2(G_k)$. In particular, $d^2=0$, defining
so {\bf cohomology groups} as the kernel of $L_k$ resp $C_k$ on $l^2(G_k)$. 
Their dimensions are the {\bf Betti numbers} of $G$. 
If $m=0$, then $d: l^2(G_k) \to l^2(G_{k+1})$ defines a standard
{\bf co-chain complex}. A particular case is if $df=\sum_k (-1)^k f(\delta_k)$, 
where $\delta_k:G_{k+1} \to G_k$ are the {\bf face maps} of a {\bf delta set} 
$G=\bigcup_{k=0}^q G_k$. 

\paragraph{}
Let us call the triple $(G,D,R)$ a {\bf finite geometry}. Its {\bf f-vector}
is defined by the cardinalities $f_k=|G_k|$. Its {\bf Betti vector} $b$ is given by
$b_k={\rm dim}({\rm ker}(C_k))$. The {\bf Euler-Poincar\'e formula}
$\sum_k (-1)^k f_k = \sum_k (-1)^k b_k$
expresses that the {\bf Euler characteristic} $\chi(G)=\sum_{x \in G} \omega(x)$ with
$\omega(x)=(-1)^{{\rm dim}(x)}$, the alternating sum of the $f_k$, agrees with
the alternating sum of the $b_k$.
The Euler-Poincar\'e identity is best shown by comparing the super trace of
$e^{-tL}$ at $t=0$ and $t=\infty$ using the {\bf McKean-Singer theorem} 
\cite{McKeanSinger,knillmckeansinger} 
that ${\rm str}(L^k)=0$ for all $k>0$ implying ${\rm str}(e^{-tL})$ 
is independent of $t \in \mathbb{R}$. 

\paragraph{}
A particular case is when $G$ is {\bf  finite abstract simplicial complex}, a finite
set $G$ of non-empty sets $x$ closed under the operation of taking finite non-empty 
subsets. In that case, $D$ and $R(x)=|x|-1$ are canonically defined. This is special
however. For example, if $G$ is the set of intersecting pairs of simplices 
in a simplicial complex, we deal with the chain complex for 
{\bf quadratic cohomology}. An other example in one-dimension are {\bf quivers}, one
dimensional delta sets which produce finite geometries that are not simplicial complexes but
for which there still is a natural Dirac matrix $D$ having the property that $D^2=L_0\oplus L_1$ is
block diagonal with $0$-form Laplacian $L_0$ and $1$-form Laplacian $L_1$. That $L_0,L_1$
are essentially isospectral (which is a special case of McKean-Singer symmetry) has been 
exploited by Anderson and Morley \cite{AndersonMorley1985} 
to estimate the spectral radius of $K_0$ leading to eigenvalue general bounds for all eigenvalues
\cite{Knill2024}. 

\paragraph{}
A bit more general than the case of simplicial complexes is if
$d$ is the exterior derivative, coming from face maps of a
{\bf delta set}. An other example is the cohomology of {\bf Wu characteristic} 
$\chi_2(G)=\sum_{(x,y) \in G} \omega(x) \omega(y)$, leading to topological invariants
but no homotopy invariants. Already for chain complexes of higher characteristics, we 
have difficulty to frame this within the category of delta sets. 
\footnote{Delta sets technically require exactly $n+1$ face maps 
from $G_n \to G_{n-1}$. We prefer with $D$ and forget about face maps.}
The formula $\chi_2(G)=\chi(G)-\chi(\delta G)$ for a manifold with 
boundary $\delta G$ illustrates why quadratic cohomology is not a homotopy invariant.
An other case going beyond the familiar simplicial complex or delta set geometry 
is obtained when the Dirac matrix $D$ is deformed in an isospectral way. We should
also mention the {\bf Witten deformation} \cite{Witten1982,Cycon} 
$d_t=e^{-t f} d e^{tf}$ which is 
not isospectral but preserves cohomology. 
Witten deformation is not part of the symmetry considered here. 

\paragraph{}
Finite geometries contain the {\bf topos of finite sets}. This is the case when $R$ is 
constant $0$ and where the matrix $D$ is the zero matrix. 
An other example is the {\bf topos of finite quivers}, 
which is the case if $R$ takes values in $\{0,1\}$ and $V=G_0$ or $E=G_1$ are the set of 
{\bf vertices} or {\bf edges} including loops. 
Quite general already is the {\bf topos of delta sets} and in particular, 
the {\bf topos of simplicial sets}. The later are delta sets with additional degeneracy maps. 
Delta sets are important because they allow the formation of products, quotients and level sets
within the category. The most general of these topoi is the topos of delta sets but 
finite geometries go even beyond that, as we do not need the matrix $D$ 
to come from {\bf face maps}. 
Looking at {\bf data structure} $(G,D,R)$ rather than insisting it 
to emerge as a pre-sheaf, is a computer science approach and can be implemented fast.

\paragraph{}
A rich source of examples for finite geometries are given by {\bf finite simple graphs} 
equipped with a {\bf Whitney complex} $G$ (given by the set of complete subgraphs). 
A finite simple graph containing a sub-graph $K_{q+1}$ but no 
$K_{q+2}$ so defines a finite $q$-dimensional geometry with dimension function 
$R(x)=|x|-1$, taking values from $0$ to $q$. Simplicial complexes are too narrow because 
there is no associative Cartesian product within simplicial complexes satisfying the 
Kuenneth formula and preserving discrete manifolds. 
The Stanley-Reisner product satisfies the later two properties 
but it is not associative as multiplication with $1$ products the Barycentric refinement. 
Finite geometries defined by finite graphs are almost as general as the abstract frame work. 
We can always for a given finite delta set look at $G$ as the vertex set of a graph and then
connect two points $x,y \in G$ if there is a sequence of face maps getting from 
$x$ to $y$ or from $y$ to $x$. 

\paragraph{}
A finite geometry $G$ can be used to define a {\bf spectral triple} $(H,D,A)$, where
$H=l^2(G)$ is a Hilbert space and $A$ is a not necessarily commutative sub Banach
algebra of $\mathcal{B}(H)$ like $l^{\infty}(G)$, the set of functions on $G$ with supremum norm.
The {\bf pseudo metric} $d(x,y)=\sup_{f, |[D,f]|_{\infty}=1} |f(x)-f(y)|$ \cite{Connes} 
defines then a {\bf metric space} on the {\bf Kolmogorov quotient}, where equivalence classes 
are points with 0 pseudo distance. This picture shows somehow plays the role of
a Riemannian metric in the continuum. The simple observation that distance alone determines
the metric tensor $g$ in a Riemannian manifold is known much longer and probably first done by
Schr\"odinger \cite{SchroedingerSpaceTimeStructure}. Spectral triples allow to 
work with discrete or non-commutative settings or both.

\paragraph{}
If the Dirac matrix $D$ comes from face maps like in a delta set, there is 
a {\bf partial order} $x \leq y$, if there is a combination of face maps getting 
from $y$ to $x$. In a simplicial complex $G$, where elements are sets of sets, 
this {\bf poset structure} is already given by inclusion $x \subset y$. 
A partial order defines then a {\bf Alexandrov topology} \cite{Alexandroff1937}
$\mathcal{O}$ on $G$: the basis for the topology 
is the set of {\bf stars} $U(x) = \{ y \in G, x \subset y\}$ which are the smallest open sets
containing $x$. In positive dimensions, this topology is never Hausdorff. 
There is {\bf arithmetic} on the entire category of finite geometries.
Start with the monoid formed by disjoint union, extend it to a ring
where the addition is first Grothendieck completed to a group, then use
Cartesian product with $R((x,y))=x+y$ and $D(G+H)=D(G) \oplus D(H)$ and 
$D*H = D(G) \otimes D(H)$, where the tensor products for $D*H$ and $H*D$ are identified.

\section{The symmetry space} 

\paragraph{}
{\bf Definition}: The {\bf symmetry space} of a geometry $(G,D,R)$ is defined as 
the set of orthogonal matrices $Q \in SO(n)$ such that $D'=Q^* D Q$ is again a Dirac 
matrix producing a finite geometry $(G,D',R)$ which is chain-homotop.
In particular, the deformed Laplacian $L'=D'^2$ has the same block diagonal structure 
$L' = \oplus_{k=0}^q L_k'$ than $L=\oplus_{k=0}^q L_k$ with 
$L_k: l^2(G_k) \to l^2(G_k)$ where $G_k=R^{-1}(k)$ and the dimensions of the kernels
of $L_k$ and $L_k'$ agree. The orthogonal $Q$ does not need to preserve $k$-forms
$l^2(G_k)$. The deformed $D'=d'_t + (d'_t)^*$ still has $d_t^2=(d_t^*)^2=0$ but 
$d_t$ now map $l^2(G_k)$ to $l^2(G_k) \times l^2(G_{k+1})$. 

\paragraph{}
The assumption implies that different geometries in the same symmetry space
have the same cohomology and Euler characteristic. 
The terminology "symmetry space" and not "symmetry group" is chosen because
this is a priori not a group, at least not with the induced group structure from $SO(n)$. 
\footnote{There is a non-trivial group structure on it provided by the isospectral flows.
It is a non-compact group of scattering paths converging to geometries that
are block diagonal. The Liouville-Arnold theory \cite{Arnold1980} suggests that that it has 
connected components that are of cylinder type $\mathbb{T}^r \times \mathbb{R}^s$,
with $r$ being non-zero if one looks a evolutions allowing complex matrices. }
To see that $S$ is not a subgroup of $O(n)$, note that 
already $(Q^2)^2 D Q^2$ is not a Dirac matrix in general because 
$Q^2$ is not necessarily block tri-diagonal. 

\paragraph{}
The {\bf isospectral Lax deformation} $D'=[B,D]$, where $B=g(D)^+-g(D)^-$ produces an isospectral
deformation $D_t=Q^*_t D Q_t$ with deformed exterior derivatives $d_t=Q^*_t d_0 Q_t$.
But unlike $d_0$ which maps $k$-forms to $(k+1)$-forms, the image of the map $d_t$ consists
also of $k$-forms. We have $d_t=c_t+n_t$, where $d_t$ does map $k$-forms to $(k+1)$-forms,
leading to $D_t= d_t + d_t^* = c_t +c_t^* + m_t$. If we look at a geometry, we only use
measurements with electromagnetic waves and only use $c_t$, not involving $m_t$. Since
we don't see it, we have called it "dark matter part" of the geometry. 
For the simplest flow $g(D)=D$, we have 
$L'=(D D)' = D' D + D D' = (B D - D B) D + D ( B D - D B) = B L- LB=[B,L]$. 
If $B=d-d^*$ and $B^2=L$ the matrices $L$ and $B$ commute so that $L'=[B,L]=0$ and
$L$ does not change.  However, already for $B=g(D)=D^3$, the Laplacian 
will be deformed. In the case when $g(D)$ is invertible, i.e. if $g(D)=e^{h(D)}$ 
we will see that the deformation is explicit in terms of $QR$ deformation. 

\paragraph{}
The isospectral deformation symmetry is a continuous symmetry, unlike the
{\bf inner symmetries} $Q(D)=D(T)$, that come from a discrete set of automorphisms $T$ 
of $G$ and which induce symmetries. The inner automorphism symmetry is in general trivial. 
\footnote{In the periodic Toda case, which has a $Z_n$ symmetry on $C_n$,
we had observed an explicit deformation within the isospectral
set from $D$ to $D(T)$, the translated operator \cite{Kni93a,Kni93b,Kni95}.
In the real case, because $SO(n)$ does not 
contain reflections, we can have disconnected isospectral parts. 
We would especially need connections from $D$ to $D(T)$. } 
\footnote{Inner symmetries could be defined more broadly and not 
require $T$ to be an automorphism
but that $T$ is {\bf continuous map} from a Barycentric refinement to $G$
and that there is an inverse of $S$ given again by a continuous map and that $TS$ and $ST$ 
are homotop to the identity. }
Certainly, if $G$ is a simplicial complex and $T$ is an automorphism, a bijection which is
permuting elements in $G$ then $(G,D(T),R)$ is again a finite geometry in 
the same symmetry than $(G,D,R)$. 
Can we connect $D$ with $D(T)$ using isospectral flows provided the corresponding orthogonal
matrix has determinant $1$? We would need to find a function $g$
such that $e^{g(D)} =Q R$ with $Qf(x) = f(Tx)$. In the Toda case, we used time dependent
Hamiltonian systems to achieve that. 

\paragraph{}
In \cite{IsospectralDirac,IsospectralDirac2} we saw that
when starting with $D=d+d^*$, the differential equation $D'=[B,D]$ with $D_0=d-d^*$ 
produces deformed operators $c_t+c_t^* + m_t$, containing now a {\bf block diagonal part} $m_t$. 
In the actually unrelated periodic Toda case, we saw that isospectral deformations allow
to interpolate the translation symmetry on the cyclic graph. There was an explicit time
dependent isospectral deformation in the isospectral set within B\"acklund transformations.
The cochain complex defined by the deformed exterior derivative $d_t$
is homotopy equivalent to the complex defined by $d$, similarly as the Witten deformation 
\cite{Witten1982,Cycon} $d_t = e^{-tf} d e^{tf}$. 

\begin{figure}[!htpb]
\scalebox{0.5}{\includegraphics{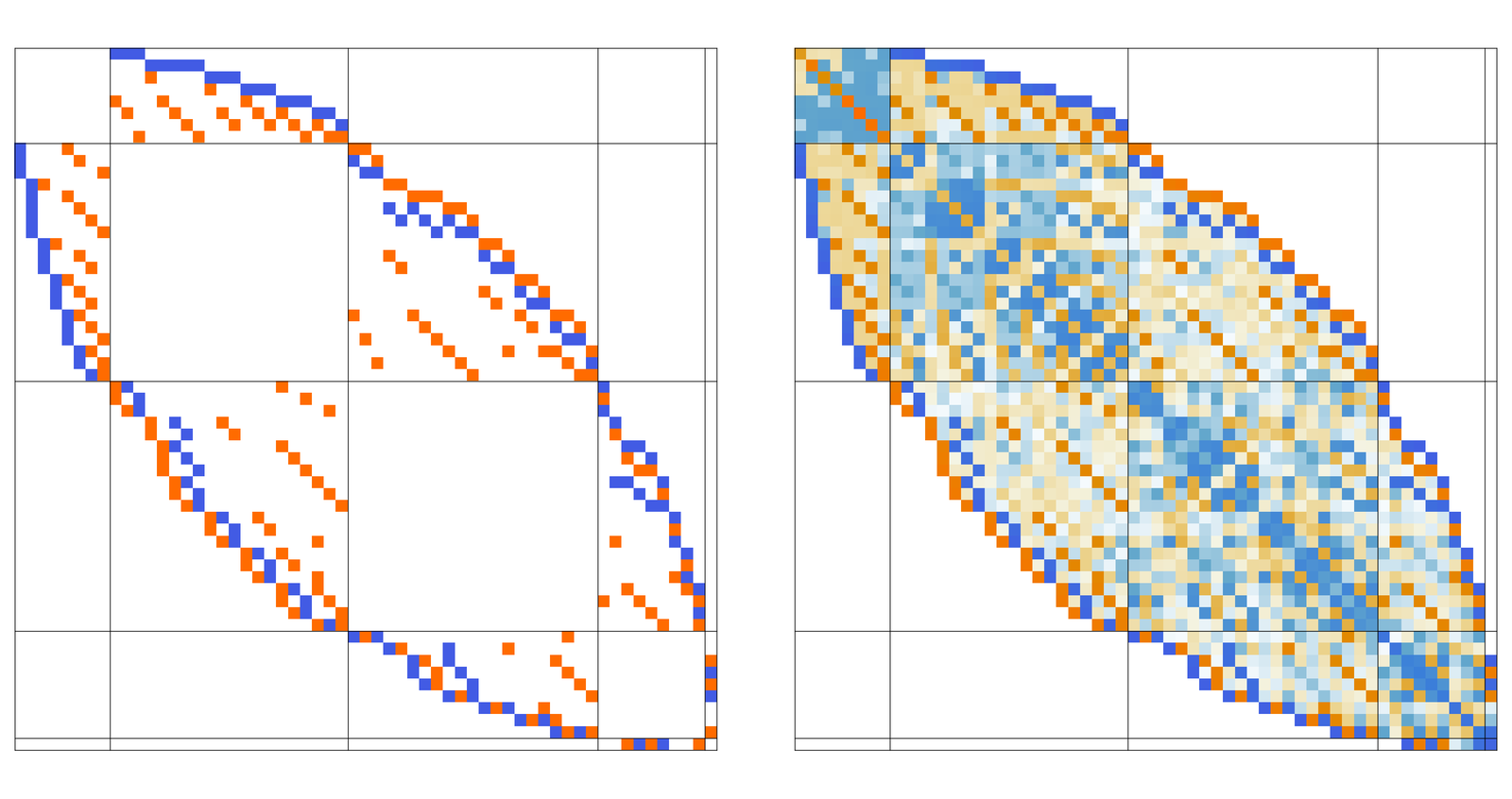}}
\label{Dirac}
\caption{
We see the Dirac matrix $D=d+d^*$ of a randomly chosen simplicial complex $G$
and the Dirac matrix $D_t=d_t+d_t^* = c_t+c_t^* + m_t$ of a deformed complex. 
Unlike the initial $D=D_0$, which has only off diagonal entries, the deformed matrix 
$D_t$ has a block diagonal part. If we focus in the new geometry of what we ``can see"
with exterior derivatives $c_t$ mapping $k$-forms to $(k+1)$-forms, 
we consider the still equivalent geometry 
$(G,C_t=c_t+c_t^*,R)$. It now represents an expanded space if we measure with the Connes
formula. 
}
\end{figure}

\paragraph{}
The isospectral deformation of the Dirac operator is motivated by the {\bf Toda chain},
\cite{Toda1967,Toda}, which is
an isospectral deformation of Jacobi matrices $L$. It is a {\bf Lax pair} \cite{Lax1968}
$L'=[B,L]$ with $B=a-a^*$ and tri-diagonal $L=a+a^*+b$. If $L$ is invertible, it is possible 
to write $L\oplus L_1=D^2$, where $D$ is defined on a doubled lattice and $L_1$ is a 
B\"acklund transformation of $L$. 
\footnote{Most integrable continuous time systems are equivalent to Lax pairs.
An example is the oscillator $L'=[B,L]$ with $L$
a reflection-dilation matrix and $B$ a rotation by $\pi/2$ \cite{Flaschka1974}. 
Also the free top $L'=[B,L]$ for angular momentum $L$.}
However, the deformation of the Dirac matrix $D$ coming from a finite geometry 
is different from Toda. It does not generalize Toda and applies
for general geometries \cite{IsospectralDirac,IsospectralDirac2}. 
It is defined here even for all finite geometries, including the topos of delta sets.
Toda on the other hand is just the deformation of a Schr\"odinger operator $L$ 
on a one-dimensional circular geometry. Deformations of higher dimensional Laplacians
are in general not possible. There is spectral rigidity in general in the sense that
Laplacians, isospectral to a given operator, in general form a discrete set. 

\paragraph{}
Let us elaborate a bit more why the geometry deformation of the Dirac matrix is different from 
Toda, even so also in the Toda case we could write a Jacobi matrix as $L=D^2$.  
Given a Dirac $D$ matrix from an arbitrary finite geometry $G$ and a polynomial $f$, then 
$f(D)$ for any polynomial is again of the same tri-diagonal type. 
This obviously is not the case for Jacobi matrices, where we have a Schroedinger 
operator on the graph $C_n$ (periodic Toda lattice) or the linear path graph $L_n$ (a
scattering situation). 
In the Dirac deformation, we have {\bf block tri-diagonal} matrices defined
by exterior derivatives on a general geometry while isospectral deformation of higher
dimensional Laplacians are in general {\bf not possible} within differential operators. 
The Lax pair for Dirac matrices is also defined for any continuum geometry, like a Riemannian 
manifold but the deformed operators are pseudo differential operators. We consider
here only the finite case. 

\paragraph{}
The next observation parallels the observation known since a long time for the Toda case 
\cite{DeiftLiTomei,Symes}: the Toda lattice flow is related to a QR flow. 
A function $g$ defines from $e^{-t g(L)}=QR$ 
a deformation $Q_t^*D Q_t=D_t$. This produces a solution of the Lax equation 
$D'=[g(D)^+-g(D)^-,D]$. Any Lax deformation of a Dirac operator is interpolated by a QR flow.

\begin{thm}
If $\exp(-t g(D_0))=Q_t R_t$, then $D_t=Q^*_t D Q_t$ solves $D' = [B,D]$,
with $B=g(D)^+-g(D)^-$. 
\end{thm}

\begin{proof} 
We verify the equivalent statement (switch $t$ to $-t$) 
that $D'=[g(D)^--g(D)^+,D]$ is realized with 
$D_t=Q^* D Q_t$, where $A_t=e^{t g(D_0)}= Q_t R_t$.
Given a path $D_t$ defined by $D_t=Q^*_t D Q_t$ with $A_t=e^{t g(D_0)}=Q_tR_t$. 
We need to show that the $Q_t$ from this decomposition satisfies the differential 
equation Then $Q'=Q (g(D_t)^--g(D_t)^+$. \\
If $Q_t$ is given by the QR-decomposition, Then $Q_t' = Q_t B_t$
with anti-symmetric $B_t$. We need to verify that $B_t$ is of the form $g(D)^--g(D)^+$.
The uniqueness of solutions of differential equations then assures that the 
QR flow and the Toda flow give the same orbits. \\
Since $A_t$ is block-triangular also $Q_t$ is block triangular as one can see when 
doing the Gram-Schmidt process. The key is to look at the differential equations
$$   R'_t=[2g(D_t)^+ + g(D_t)^0] R_t,   Q'_t=Q_t [g(D_t)^--g(D_t)^+] $$
for the pair $R,Q$ with $R(0)=Q(0)=1$ coming from $e^{-0 g(D)}=1 \cdot 1$.
The $R_t$ remains upper triangular. 
Now differentiate $A_t=Q_t,R_t$: 
$$   g(D_0) A = A'=Q'R+QR' = Q_t [g(D_t)^--g(D_t)^+] R_t + Q_t [2g(D_t)^+ + g(D_t)^0 ] R_t
                       = Q_t g(D_t) R_t \; . $$
Because $e^{t g(D_0)}$ is invertible for all $t$, also $R_t$ is invertible
for all $t$ and the equation $g(D_0) Q_tR_t = Q_t g(D_t) R_t$ is 
equivalent to $g(D_0) Q_t = Q_t g(D_t)$.
Also $Q_t^* D_0 Q_t = D_t$ implies $Q_t^* g(D_0) Q_t = g(D_t)$. \\
Having established that $Q'_t = Q B_t$, we see that 
$D_t' = -B_t Q_t + Q_t B_t$ matching $E_t'=[g(D_t)^--g(D_t)^+,D_t]$ 
Because this ordinary differential equation in the
matrix algebra has global solutions $B_t = g(D_t)^--g(D_t)^+$. 
\footnote{There are global bounds on the norm of 
$B_t$ and $Q_t$ which by the way do not hold in infinite dimensions \cite{DeiftLiTomei}.}
The statement in the theorem is the version with $t$ replaced with $-t$. 
\end{proof} 

\paragraph{}
All these flows commute. The verification is the same as in the Toda case. 
See \cite{Symes}. And $e^{-tg(D)}$ is independent of $t$ if and only if $g(D)=0$. This means
that we have no equilibria for all$t$. We have to be careful however because the flow 
extends in the limit $t \to \infty$ to situations $D=c+c^* + m = m$,
where $c=0$ which are not geometric any more and which we do not consider to be
equivalent. 

\paragraph{}
The McKean Singer formula still holds also in the deformed case. The 
super trace any power $L_t^k$ of $L_t=D_t^2$ is still zero implying
${\rm str}(e^{-L_t}) = \chi(G)$. 

\paragraph{}
The deformation can happen in over the complex field $\mathbb{C}$. To get 
complex solutions, look at 
$$ D' = [ g(D)^+- \overline{g(D)}^-+ i \beta g(D)^0, D]   $$
Now that at $t=0$ and $g(D)=D$, we have $g(D)^0=0$ but still 
get an evolution. One can still associate with this with some sort of complex
 $QR$ decomposition by modifying the differential equations for $R,Q$
$$   R'_t=    [2g(D_t)^+ + g(D_t)^0 - i \beta g(D)^0] R_t ,   
     Q'_t=Q_t [\overline{g(D_t)}^--g(D_t)^+ + i \beta g(D)^0   ] $$
Now, $Q(t)$ is unitary and $R(t)$ block upper triangular complex. 
In the limit $t \to \infty$, we reach a Schr\"odinger 
wave evolution $Q_t' = i m_t \beta Q_t$. This complex
evolution can also be pushed to quaternions by replacing $i$ with a
unit quaternion. 

\paragraph{}
With $g(D)=\log(1+c D)$, which is defined for small $c$, 
we get $(1+c D)^t = Q_t R_t$ which produces
a discrete time evolution when restricting $t$ to the integers. Now,
take $t=m$, an integer then  $(1+c D)^m$ is a polynomial and if $x$
is a simplex, look at $(1+c D)^m = Q_m R_m$.
It produces a deformation $D_m = Q_m^* D_0 Q_m$ which has the property that properties
of the geometry in distance larger than $m$ to $x$ do not influence the motion 
of $D$ at $x$.  We can make $c$ time dependent 
for fixed $t$, and could look at the polynomial $(1-t D/k)^k \to e^{-t D}$. 
We see that we can approximate the deformation path with an orbit of a
{\bf discrete time system} that has the {\bf local property} simplices $y$
of $G$ in distance larger than some distance $L$ are not affecting
the change of entries like $D_{x,x}$ 

\section{Discussion}

\paragraph{}
Lets start with a more historical or philosophical remark. 
Geometers at the time of Euclid looked at symmetries like 
{\bf similarities} or {\bf congruences} of triangles or circles. Klein's Erlanger program
saw symmetries in the form of {\bf symmetry groups} defining the class of geometry 
under considertation. N\"other related symmetries with {\bf conservation laws}.
The special covariance principle in the form of Lorentz or Poincar\'e groups
guided special relativity. Representations of the Poincar\'e group are "particles". 
The general covariance principle related to the diffeomorphism group of a manifold 
is central to general relativity. Arnold \cite{Arnold1980} noticed, motivated from the fact
that the free motion of a mass point is a geodesic in the rotation group, that
fluid flows like the Euler equations are {\bf geodesics in the diffeomorphism group}.
Such principles suggest the following postulate:
{\bf a geometry is allowed to deform freely along geodesics in its symmetry group.}
In the geometric frame work considered here, where $D_t$ converge to block diagonal 
case where $c=0$ and so $m = d+d^*$, meaning that each $d$ preserves the set of $k$-forms. 
A consequence is that "finite geometries in general spontaneously expand when 
distances are measured using the electromagnetic exterior derivatives $c_t$.

\paragraph{}
That a geometry can float in its symmetry space is usually not
spectacular. The reason is that physical properties of the geometry stay the same. 
If our space-time manifold would undergo a diffeomorphism change, then by the 
{\bf equivalence principle}, all physical properties remain the same
despite that the change of coordinates produces forces but these forces are a 
consequence of the equivalence principle. A rigid body released in space, 
if free from any external forces, will in general rotate like a free top. This is 
described by a Lax pair $L'=[\omega,L]$. It is integrable and produces
geodesics in $SO(3)$. The probability that it 
does not rotate or move with respect to a reference frame is zero. One could
think to go into a coordinate system where centrifugal and Coriolis forces 
are minimal, but such considerations are non-relativistic.
The deformed Dirac matrix $D_t=d_t+d_t^* + m_t$ is 
a {\bf change of coordinates} which relates it to $D_0 = d+d^*$. It is completely 
equivalent to the original $D_0$ as it is given by a linear change of coordinates.
But the restriction of the Dirac matrix to its (electro-magnetic) exterior part
makes the geometry to expand we we measure with electro-magnetic tools (mathematically 
given by exterior derivatives $c_t$ that map $k$-forms to $(k+1)$-forms.
The association of electro-magnetism is not that far-fetched because all distance 
measurements we are aware of eventually boil down to electro-magnetic forces and so 
the wave equation $u_{tt} = - L u$ for the Laplacian $L=(c_t+c_t^*)^2$ of the geometry. 

\paragraph{}
From a mathematical perspective, if we look at the motion of the exterior derivative
$c_t$ with $C_t=c_t+c_t^*$ alone and look the Connes formula for 
$C_t$ rather than $D_t=c_t+c_t^* + m_t$, the
geometric effect is visible, as the $c_t$ shrink in general and space expands.
The QR flow in general converges to a block diagonal matrix, where in the limiting 
case (which we never reach in finite time), 
we have no geometry any more. The $Q$ which diagonalizes $D$ is still 
in the symmetry group but the exterior derivative $c$ is zero then. This is
a limiting situation which the QR flow in general tends towards to. We expect a finite
set of points in the symmetry group to consists of matrices $D=m$ without exterior 
derivative. We can such points still as geometry since the Betti numbers defined Hodge 
theoretically as the kernels of $L_k$ are still the same and because we still have 
wave dynamics. But any Lax flow $D'=[g(D)^+-g(D)^-,D]=0$ does not move.
We have seen above that if we allow the geometry to move in the {\bf complex symmetry}
then in the limit we still move and have a wave dynamics. 

\paragraph{}
Exterior derivatives are associated to classical physics like
the Maxwell equations $dF=0, d^*F=j$ which give the electro-magnetic
field $F \in l^2(G_2)$ from the charge-current $j \in l^2(G_1)$. 
As the diagonal part $m$ in the Dirac matrix $D=c_t + c_t^* +m$ is not
electro-magnetic, it associates with dark matter speculations.
If we subscribe to the covariance principle that time evolution is just
a choice of geodesic flow in the symmetry group of a finite geometry
and assume that at time $t=0$, we ahve $m_0=0$.
During the evolution, the non-classical part $m_t$ 
grows in general in time and the $c_t$ which is used to define
waves like the wave operator $u(s) = \cos((c_t+c_t^*) s) u(0)$ solving
the wave equation $u_{ss} = -(c_t+c_t^*)^2 u$. Note that the solution $u(t)$
simultaneously looks at the deformation on $l^2(G)$ and especially on each
form sector like {\bf $1$-forms} $l^2(G_1)$ which can represent {\bf electromagnetic potentials} 
$A(s) \in l^2(G_1)$, defining electromagnetic fields $F(s)=dA(s) \in l^2(G_2)$ 
represented as {\bf 2-forms}.

\paragraph{}
To summarize: we have seen that if a geometry given initially as $D=d+d^*$ is allowed to 
float freely in its symmetry, it in general expands. If we look at
the electro-magnetic $c+c^*$ part of $D_t=c_t+c_t^*+m$, then
what that we focus on what we can ``see" using light. 
\footnote{Light is based on the Maxwell equations $dF=0,d^*F=j$ assume that $F=dA$ 
is a 2-form and $j$ is a 1-form. In the deformed $D_t=d_t+d_t^*$
the $d_t A$ is a mixture of 1-forms and 2-forms. }
With $c_t+c_t^*$ and neglecting $m_t$, we get a traditional geometry, where 
the exterior derivative maps $k$-forms to $(k+1)$-forms. But distances
given by Connes formula have now expanded. Also interesting is that we can find
discrete time evolutions in the symmetry which are local, meaning that signals
propagate with {\bf finite speed}. This is not the case for the wave equation 
$u_{tt} = -L u$ (understanding having continuous time and discrete space).
The solution $u_t = \sum_k a_k \cos(\sqrt{\lambda_k} t) \psi_k$ 
for the initial $u_0=\sum_k a_k \psi_k$ decomposed in an eigen-system $\psi_k$
uses eigen-functions $\psi_k$ that depend on the entire geometry. The value $u_t(x)$ can depend 
on $u_0(y)$ with $y$ located arbitrarily far away. It is only in the continuum limit that we
have strict finite propagation speed of the wave equation. We have once modified the wave
equation for discrete operators so that we have finite propagation speed \cite{Kni98}.

\section{Code}

\paragraph{}
The following Mathematica code allows to compute both the QR flow as well
as the Lax deformation for an arbitrary simplicial complex $G$ and compares 
how close we are. 
\footnote{Mathematica 14.1.0 produces sometimes different sign in the last diagonal entry 
of $R$ in the QR decomposition. We corrected this in August 27, 2024 and programmed 
the QR decomposition by hand.}

\begin{tiny}
\lstset{language=Mathematica} \lstset{frameround=fttt}
\begin{lstlisting}[frame=single]
Generate[A_]:=If[A=={},{},Sort[Delete[Union[Sort[Flatten[Map[Subsets,A],1]]],1]]];
Whitney[s_]:=Union[Sort[Map[Sort,Generate[FindClique[s,Infinity,All]]]]]; L=Length;
sig[x_]:=Signature[x];   nu[A_]:=If[A=={},0,L[A]-MatrixRank[A]];  omega[x_]:=(-1)^(L[x]-1);
F[G_]:=Module[{l=Map[L,G]},If[G=={},{},Table[Sum[If[l[[j]]==k,1,0],{j,L[l]}],{k,Max[l]}]]];
sig[x_,y_]:=If[SubsetQ[x,y]&&(L[x]==L[y]+1),sig[Prepend[y,Complement[x,y][[1]]]]*sig[x],0];
Dirac[G_]:=Module[{f=F[G],b,d,n=L[G]},b=Prepend[Table[Sum[f[[l]],{l,k}],{k,L[f]}],0];
   d=Table[sig[G[[i]],G[[j]]],{i,n},{j,n}];{d+Transpose[d],b}];
Hodge[G_]:=Module[{Q,b,H},  {Q,b}=Dirac[G];  H=Q.Q;
   Table[Table[H[[b[[k]]+i,b[[k]]+j]],{i,b[[k+1]]-b[[k]]},{j,b[[k+1]]-b[[k]]}],{k,L[b]-1}]];
Betti[s_]:=Module[{G},If[GraphQ[s],G=Whitney[s],G=s];Map[nu,Hodge[G]]];
Fvector[A_]:=Delete[BinCounts[Map[L,A]],1];  Euler[A_]:=Sum[omega[A[[k]]],{k,L[A]}];
QR[A_]:=Module[{F,B,n,T},T=Transpose;F=T[A];n[x_]:=x/Sqrt[x.x];B={n[F[[1]]]};Do[v=F[[k]];
 u=v-Sum[(v.B[[j]])*B[[j]],{j,k-1}];B=Append[B,n[u]],{k,2,Length[F]}];{T[B],B.A}];
QRDeformation[B_,t_]:=Module[{Q,R,EE},{Q,R}=QR[MatrixExp[-t*B]]; Transpose[Q].B.Q];
PS={ColorFunctionScaling -> False};
DiracPlot[{B_,b_}]:=Module[{S1,S2,S3},n=L[B];S1=MatrixPlot[B,FrameTicks->None,PS,Frame->False];
S2=Graphics[{Thickness[0.001],Table[Line[{{0,n-b[[k]]},{n,n-b[[k]]}}],{k,L[b]}]}];
S3=Graphics[{Thickness[0.001],Table[Line[{{b[[k]],0},{b[[k]],n}}],{k,L[b]}]}];Show[{S1,S2,S3}]];
LowerT[A_]:=Table[If[i>=j,0,A[[i,j]]],{i,Length[A]},{j,Length[A[[1]]]}];
Str[B_,b_]:=Module[{},Sum[-(-1)^k*Sum[B[[b[[k]]+l,b[[k]]+l]],{l,b[[k+1]]-b[[k]]}],{k,L[b]-1}]];
RK[f_,x_,s_]:=Module[{u,v,w,q},u=s*f[x];v=s*f[x+u/2];w=s*f[x+v/2];q=s*f[x+w];x+(u+2v+2w+q)/6];
DiracDeformation[DD_,tt_]:=Module[{dt=1./10^5,d,e,B,q=DD,T}, NN=Floor[tt/dt];
   Do[d=LowerT[q];e=Transpose[d];B=e-d;FF[x_]:=B.x-x.B;q=RK[FF,1.*q,dt],{NN}];q];
ComplexDirac[DD_,tt_]:=Module[{dt=1./10^5,d,e,m,B,q=DD}, NN=Floor[tt/dt];Do[d=LowerT[q];
  e=Conjugate[Transpose[d]];m=q-d-e;B=e-d+I m;FF[x_]:=B.x-x.B;q=RK[FF,1.*q,dt],{NN}];q];
EvolvDirac[DD_]:=Module[{dt=1./10^5,d,e,m,B,q=DD,Q={}},Do[d=LowerT[q];e=Conjugate[Transpose[d]];
  m=q-d-e;B=e-d+0*I m;FF[x_]:=B.x-x.B;q=RK[FF,1.*q,dt];Q=Append[Q,Tr[m.m]],{Floor[2.0/dt]}];Q];

s=RandomGraph[{8,20}]; G=Whitney[s];{B,b}=Simplify[Dirac[G]]; n=Length[B]; 
V[X_]:=Max[Abs[Flatten[Table[X[[k,l]],{k,n},{l,n}]]]];
B2=DiracDeformation[B,0.12]; B1=QRDeformation[B,0.12]; Print["Difference: ",V[B2-B1]];
S=GraphicsGrid[{{DiracPlot[{Chop[B],b}],DiracPlot[{Chop[B1],b}]}}]

s=CompleteGraph[4]; G=Whitney[s]; {B,b}=Simplify[Dirac[G]];n=Length[B];A=EvolvDirac[B];
S=ListPlot[Table[Re[A[[k+1]]-A[[k]]],{k,Length[A]-1}],Filling->Bottom,FillingStyle->Yellow]
\end{lstlisting}
\end{tiny}

\begin{figure}[!htpb]
\scalebox{0.9}{\includegraphics{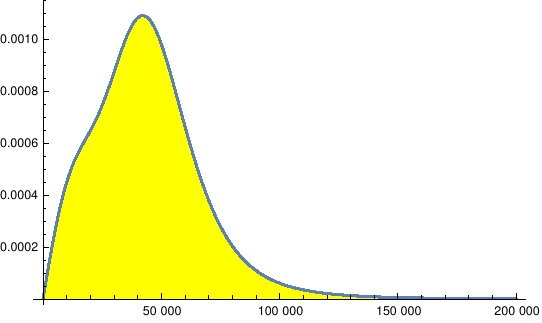}}
\label{inflation}
\caption{
The rate of change of the norm $||m(t)||$ of the block diagonal 
matrix $m(t) = \oplus_{k=0}^q m_k(t)$ acting on 
$l^2(G)=\oplus_{k=0}^q l_2(G_k)$ shows inflation. The graph looks  
pretty similar for different $G$. 
}
\end{figure}

\bibliographystyle{plain}

\end{document}